\documentclass[%
12pt%
,a4paper%
]{article}

\usepackage{tikz}
\usepackage{amssymb}
\usepackage{amsmath}
\usepackage{amsthm}
\usepackage[utf8]{inputenc}
\usepackage[english]{babel}
\usepackage{graphicx}

\usepackage[a4paper,margin=2.5cm]{geometry}

\usepackage{hyperref}

\theoremstyle{plain}
\newtheorem{theorem}{Theorem}[section]
\newtheorem{proposition}[theorem]{Proposition}
\newtheorem{lemma}[theorem]{Lemma}

\theoremstyle{definition}

\newtheorem{assumption}[theorem]{Assumption}
\newtheorem{remark}[theorem]{Remark}

\numberwithin{equation}{section}

\newcommand{\abs}[1]{\lvert{#1}\rvert}
\newcommand{\norm}[1]{\lVert{#1}\rVert}
\newcommand{\ket}[1]{\lvert{#1}\rangle}
\newcommand{\bra}[1]{\langle{#1}\rvert} 
\newcommand{\ip}[2]{\langle{#1}|{#2}\rangle}

\DeclareMathOperator{\im}{Im}
\DeclareMathOperator{\re}{Re}

\begin{document}
\title{On the adiabatic theorem when eigenvalues dive into the continuum.}
\author{
H. D. Cornean\footnote{Department of Mathematical Sciences, Aalborg University, Fredrik Bajers Vej 7G, 9220 Aalborg \O, Denmark. E-mail: cornean@math.aau.dk, matarne@math.aau.dk, hanskk@math.aau.dk }, \;
A. Jensen\footnotemark[1], \;
H. K. Kn\"orr\footnotemark[1], \;
and G. Nenciu\footnote{Institute of Mathematics of the Romanian Academy, Research Unit 1, P.O. Box 1-764, 014700 Bucharest, Romania. E-mail: gheorghe.nenciu@imar.ro} 
}
\date{}

\maketitle

\begin{abstract}
We consider a reduced two-channel model of an atom consisting of a quantum dot coupled to an open scattering channel described by a three-dimensional Laplacian. We are interested in the survival probability of a bound state when the dot energy varies smoothly and adiabatically in time. The initial state corresponds to a discrete eigenvalue which dives into the continuous spectrum and re-emerges from it as the dot energy is varied in time and finally returns to its initial value. Our main result is that for a large class of couplings, the survival probability of this bound state vanishes in the adiabatic limit. At the end of the paper we present a short outlook on how our method may be extended to cover other classes of Hamiltonians; details will be given elsewhere.
\end{abstract}


\section{Introduction}

We consider a reduced two-channel model of an atom consisting of a zero-dimensional quantum dot coupled to an open scattering channel described by a three-dimensional Laplacian. The energy level $E$ of the quantum dot varies smoothly and adiabatically in time. Initially the atom is in a bound state. If the bound state does not dive into the continuous spectrum during the variation of the dot energy $E$, the standard adiabatic theorem yields that the survival probability of the bound state (at the end of the variation of $E$) is one in the adiabatic limit.

The main problem we address here is what happens with the survival probability of the bound  state when the bound state dives into the continuous spectrum for a while during the adiabatic tuning of $E$. 
On the one hand the survival probability still remains one in the adiabatic limit by the gapless adiabatic theorem \cite{AE,T1,T2} if the bound state persists in the continuum and the corresponding spectral projection is twice differentiable.

On the other hand if the bound state turns into a resonance (in the sense that the analytically continued resolvent has a pole) due to the coupling between the atom and the open channel, the general belief is that the survival probability goes to zero in the adiabatic limit. The heuristics behind this is that the particle makes transitions to continuum states and then scatters towards infinity. In particular, this is the heuristic argument for the existence of the so-called spontaneous (adiabatic) pair creation in linear quantum electrodynamics \cite{N,PD,Th}, as well as for memory effects in quantum mesoscopic transport \cite{BP,CJN,HC,LC}.

Turning this heuristics into a mathematical statement proved to be a hard problem and boiled down to a proof of various aspects of the adiabatic theorem in the case where the eigenvalue hits the threshold of the continuous spectrum. Accordingly, the existence results are very limited and the proofs are rather technical and often need further assumptions \cite{CJN,N,PD}.

Our main result states that the survival probability vanishes in the adiabatic limit, i.e. the adiabatic theorem breaks down, for a large class of couplings between the quantum dot and the open channel, when the bound state dives into the continuous spectrum during the adiabatic tuning of $E$. In addition, a detailed spectral analysis of the model is given and a `threshold adiabatic theorem' is proved.

Our model is considerably simpler than the one in \cite{PD} and allows for rather straightforward dispersive estimates, while the threshold analysis is self-contained. In spite of the relative simplicity of the model, our method is quite robust and may be generalized to cover a larger class of operators. 
In Section~\ref{sec: Dirac} we present a short outlook on the Dirac and $N$-body Schr\"odinger case with $N\geq 2$. Details will be given elsewhere. 

Below we present the setting, formulate the results, and comment on them. Sections~\ref{spectral-analysis}--\ref{subsec: oc-stable} contain the proofs of the statements in the main theorem.

\subsubsection*{Setting and results}\label{sec: model}
The Hilbert space of the model is $\mathcal{H} = L^2(\mathbb{R}^3) \oplus \mathbb{C}$. 
Without a coupling the electron of the atom can either move freely in the open channel $\mathbb{R}^3$ or be localized in the quantum dot with dot energy $E\in \mathbb{R}$. 
In matrix notation the coupled Hamiltonian is given by
\begin{align}\label{eq: Hamiltonian}
H_{\tau}(E) = \begin{bmatrix} -\Delta & 0\\ 0 & E \end{bmatrix}
+ \tau\begin{bmatrix} 0 & \ket{\varphi}\\ \bra{\varphi} & 0 \end{bmatrix}
\end{align}
where $\tau$ is a real coupling parameter and $\varphi \in L^2(\mathbb{R}^3)$ is the coupling function. 
A normalized basis vector of the quantum dot is denoted by $\varsigma$. An alternative representation of the Hamiltonian is
\begin{align*}
H_{\tau}(E) = -\Delta + E \ket{\varsigma}\bra{\varsigma}+\tau \bigl( \ket{\varphi}\bra{\varsigma} + \ket{\varsigma}\bra{\varphi} \bigr).
\end{align*}
Here and in the sequel we use the Dirac notation with the usual ambiguities this implies. For example $\ket{\varsigma}$ denotes both a basis of the subspace $\{0\}\oplus\mathbb{C}$ of $\mathcal{H}$ and a map from $\mathbb{C}$ to $\mathcal{H}$.\\

We recall the definition of the usual weighted $L^2$-spaces
\begin{equation*}
L^{2,w} ( \mathbb{R}^3 ) = \{ f \in L^2(\mathbb{R}^3) \,|\, \int_{\mathbb{R}^3} ( 1 + x^2 )^{w} \abs{f(x)}^2 dx < \infty\},
\quad w\in\mathbb{R}.
\end{equation*}
We write $\mathcal{F}(f)=\widehat{f}$ for the Fourier transform of $f \in L^2(\mathbb{R}^3)$ and $\mathcal{F}^{-1}(f)=\check{f}$ for its inverse. The coupling function $\varphi$ has to fulfill the following conditions: 
\begin{assumption}\label{cond: coupling 1}
Let $\varphi\in L^{2,w}(\mathbb{R}^3)$ for all $w>0$ with $\norm{\varphi}_{L^2(\mathbb{R}^3)} = 1$.
Assume that there exists some integer $\nu\geq1$ such that $\abs{k}^{-\nu}\widehat{\varphi}(k)$ is continuous at $k=0$.
\end{assumption}

Note that there is no loss of generality in assuming that the coupling function is normalized. The first condition on $\varphi$ implies that $(1+x^2)^{n}\varphi \in L^1(\mathbb{R}^3)$ for any $n \in \mathbb{N}$.
Hence $\widehat{\varphi}$ is $C^\infty$ and uniformly bounded.
The continuity condition at zero implies that all derivatives of $\widehat{\varphi}$ vanish at $k=0$ for $\abs{\alpha} < \nu$ and only those derivatives of degree $\abs{\alpha} \geq \nu$ may be non-zero there.
Thus all terms of degree less than $\nu$ vanish in the Taylor expansion of $\widehat{\varphi}$ around $k=0$.
In particular, we always have $\widehat{\varphi}(0)=0$.

We study the Heisenberg evolution of the bound state when the dot energy $E$ changes adiabatically with time.
To this end we let $E$ be time-dependent, $t \mapsto E(\eta t)$, with a parameter $\eta > 0$. In this context the adiabatic limit means $\eta \downarrow 0$.

The time-dependence of $E$ should satisfy the following conditions:

\begin{assumption}\label{cond: switching 1}
The function $E \colon [-1,0] \to \mathbb{R}$ is $C^2( [ -1,0 ] )$. There exists $s_m \in ( -1,0 )$ such that $E(\cdot)$ is strictly increasing on $[-1,s_m ]$ and strictly decreasing on $[ s_m, 0 ]$. Its maximal value $E_m=E(s_m)$ is positive while $E(-1)=E(0)<0$. 
\end{assumption}

Then given any intermediate value $E \in (E(-1),E_m)$ there exist exactly two  points $s<s_m<s'$ such that $E(s)=E(s')=E$.
A sketch of a function $E(\cdot)$ satisfying the above conditions is given in Figure~\ref{fig:1}.

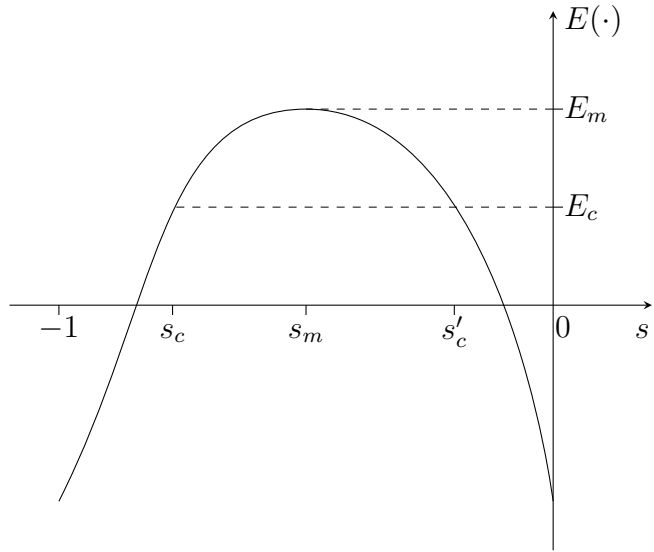
\begin{figure}[ht]
\centering
\begin{tikzpicture}[>=stealth]
\begin{scope}[yscale=1.3,xscale=1.3]
  \draw[->] (-5.5,0) -- (1,0);
  \draw (0.9,-0.08) node[below] {$s$};
  \draw[->] (0,-2.5) -- (0,3.0);
  \draw (0,2.9) node[anchor=west] {$E(\cdot)$};
  \draw (-5,0) -- (-5,-0.1);
  \draw (-5,-0.0) node[below] {$-1$};
  \draw (0.1,0.0) node[below] {$0$};
  \draw (-5,-2) .. controls (-4,0) and (-4,2) .. (-2.5,2)
  .. controls (-0.5,2) and (0,-2) .. (0,-2);
  \draw (0,1) -- (0.1,1);
  \draw (0,1) node[anchor=west] {$E_c$};
  \draw (-1,0) -- (-1,-0.1);
  \draw (-1,0) node[below] {$s_c'$};
  \draw (-3.85,0) -- (-3.85,-0.1);
  \draw (-3.85,-0.08) node[below] {$s_c$};
  \draw (-2.5,0) -- (-2.5,-0.1);
  \draw (-2.5,-0.08) node[below] {$s_m$};
  \draw (0,2) -- (0.1,2);
  \draw (0,2) node[anchor=west] {$E_m$};
  \draw[dashed,very thin] (0,2) -- (-2.5,2);
  \draw[dashed,very thin] (0,1) -- (-3.85,1);
\end{scope}
\end{tikzpicture}
\caption{Sketch of a generic $E(\cdot)$ where $E_c<E_m$}\label{fig:1}
\end{figure}

We shall perform a detailed spectral analysis of the instantaneous operator $H_{\tau}(E)$ in  Section~\ref{spectral-analysis}. The essential spectrum of $H_{\tau}(E)$ is $[0,\infty)$. Furthermore, it is proved that for any value of $\tau\neq0$ there exists a critical value $E_c>0$ such that the operator $H_{\tau}(E)$ has exactly one simple negative eigenvalue for every $E<E_c$.  Moreover, $H_{\tau}(E_c)$ has the eigenvalue $0$ embedded at the threshold.

If the operator $H_{\tau}(E(s))$ has an eigenvalue $\lambda(E(s))$ for a given value of $s$ we let $P(E(s))$ denote the corresponding eigenprojection and $\Psi(E(s))$ a normalized eigenfunction.
For $s$ equal to either $s_c$ or $s_c'$ we use the shorthand notation $P_c = P(E(s_c))= P(E(s_c'))$. Any `critical' eigenfunction is denoted by $\Psi_c$.

The time-evolution of $H_{\tau} (E(\eta t))$ is given by the time-dependent Schr\"o\-di\-nger equation
\begin{equation*}
i \frac{\partial}{\partial t} U_{\eta} (t,t_0) = H_{\tau} (E(\eta t)) U_{\eta} (t,t_0),\quad U_{\eta} (t_0,t_0) = {\rm Id},
\end{equation*}
for $t, t_0 \in \mathbb{R}$. Now we are prepared to state the main result of our paper:

\begin{theorem}\label{thm-horia}
Let  $\varphi$  be as in Assumption~\emph{\ref{cond: coupling 1}} and $E(\cdot)$ as in Assumption~\emph{\ref{cond: switching 1}}.
Then there exists a $\tau_0 > 0$, which depends on the choice of $\varphi$ and $E(\cdot)$, such that for every fixed $\tau\in (0, \tau_0)$ the following statements hold:
\begin{enumerate}

\item[{\rm (i)}] {\bf (Spectral analysis at the threshold) }

There is a critical dot energy $E_c\in (0,E_m)$ such that $\lambda(E_c) = 0$ is an embedded simple eigenvalue of $H_{\tau} (E_c)$ with corresponding  eigenprojection $P_c$. For every $E<E_c$ there exists a unique discrete negative eigenvalue $\lambda(E)$ corresponding to a smooth eigenprojection $P(E)$ with
\begin{align}\label{horia6}
\norm {P'(E)} \leq \frac{C}{(E_c-E)^{3/4}}\quad {\rm and}\quad \lim_{E\uparrow E_c}\norm {P(E)-P_c} =0.
\end{align}

\item[{\rm (ii)}] {\bf (An adiabatic theorem up to the threshold) }

There exist two positive constants $C$ and $\eta_0$, and an exponent $\alpha(\nu)\in [1/13,1]$ such that for any $\eta < \eta_0$,
\begin{align}\label{horia1}
\norm{ U_{\eta} ( s_c/\eta, -1/\eta ) P(E(-1)) U_{\eta}^* ( s_c/\eta, -1/\eta ) - P_c } &\leq C\eta^{\alpha(\nu)},\nonumber\\
\norm{ U_{\eta} ( s_c'/\eta, 0 ) P(E(0)) U_{\eta}^* ( s_c'/\eta, 0 ) - P_c } &\leq C\eta^{\alpha(\nu)}. 
\end{align}
If $\nu\geq 7$ then $\alpha(\nu)=1$. 

\item[{\rm (iii)}] {\bf (Zero survival probability when the instantaneous bound state becomes a `true resonance' for a while) }

There exists a class of functions $\varphi$ for which the instantaneous Hamiltonian $H_\tau(E)$ has purely absolutely continuous spectrum when $E_c<E\leq E_m$, and
\begin{align}\label{thm: main-4}
\lim_{\eta \downarrow 0} \abs{ \ip{\Psi (E(0))}{U_{\eta} ( 0, -1/\eta ) \Psi (E(-1)) } }^2 = 0.
\end{align}

\item[{\rm (iv)}] {\bf (When the adiabatic theorem holds all the way) }

Assume that $\widehat{\varphi}(k) = 0$ on $B_\delta(0)$ for some $\delta > 0$ and that $0 < E_m <  \delta^2$. 
Then $H_\tau(E(s))$ has a simple eigenvalue $\lambda(E(s))$ for every $s\in[-1,0]$. The eigenvalue is discrete and negative outside $[s_c,s_c']$ and embedded in the continuous spectrum inside $[s_c,s_c']$. Moreover,
\begin{align}\label{horia4}
\lim_{\eta \downarrow 0} \abs{ \ip{\Psi (E(0))}{U_{\eta} ( 0, -1/{\eta} ) \Psi (E(-1)) } }^2 = 1.
\end{align}

\item[{\rm (v)}]  {\bf (For which microscopic times does a critical eigenvector generally survive in the continuum?) }

There is a constant $K >0$ such that for every $\alpha >0$ there exists $\eta_0(\alpha)>0$ such that for every $\eta<\eta_0(\alpha)$ and $s \in ( s_c, s_c' )$ with $\abs{s - s_c} \leq \alpha  \eta^{\frac{4}{2\nu+7}}$,
\begin{align}\label{thm: main-3}
\abs{ \ip{\Psi_c}{U_{\eta} ( s/\eta, s_c/\eta) \Psi_c} }^2\geq 1-\alpha K.
\end{align}
\end{enumerate}
\end{theorem}
We now give remarks and comments on the theorem.
\begin{remark}
The class of functions in Theorem~\ref{thm-horia}(iii) are those satisfying  Assumption~\ref{cond: coupling 2}, in addition to 
Assumption~\ref{cond: coupling 1}.
Examples  are those that satisfy either  $\nabla \widehat{\varphi}(0) \neq 0$ or $\widehat{\varphi}(k) = e^{-\abs{k}^{-2}}$ near $k=0$, see the discussion after Assumption~\ref{cond: coupling 2}.
\end{remark}

\begin{remark}\label{horia3}
\ 
\begin{enumerate}
\item The bounds in \eqref{horia1} show that the instantaneous bound state is a good approximation to the `true' Heisenberg evolution up to the threshold. 

\item Roughly speaking, the adiabatic theorem without a gap \cite{AE, T1, T2} states that if one can define an instantaneous spectral projection which can be followed smoothly as a function of $s$ through possible crossings with other parts of the spectrum, then the instantaneous spectral projection always stays close to its Heisenberg evolution. In our case, if we choose for example $\varphi$ such that $\widehat{\varphi}(k)=\frac{k_1}{(k^2+1)^2}$ (here $\nu=1$ and $\nabla \widehat{\varphi}(0)\neq 0$) the instantaneous eigenvalue turns into a resonance with a non-zero imaginary part if $s\in (s_c,s_c')$. Of course, the instantaneous residue of the resolvent is no longer an orthogonal projection during this time. Furthermore, \eqref{thm: main-4} states that the critical eigenprojection becomes orthogonal to its adiabatic Heisenberg evolution and remains so when the bound state comes out of the continuum. 

\item The limit in~(\ref{thm: main-3}) shows that the critical eigenvector survives for a long {\bf microscopic} time if the coupling function $\widehat{\varphi}$ has a high-order zero at $k=0$. But this does not suffice to decide what happens with the state when $s=s_c'$. For example, when $\varphi \in L^{2,w} ( \mathbb{R}^3 )$ and $\widehat{\varphi}(k) = e^{-\abs{k}^{-2}}$ near $k=0$ we know that (\ref{thm: main-3}) holds for all $\nu$, but eventually the survival probability is zero. On the other hand, if $\widehat{\varphi}$ is zero on a neighborhood of $k=0$, then the critical eigenvector can reappear almost intact at  $s=s_c'$ and return to the initial state as in (iv). 

\item Our result does not really depend on the fact that the Laplacian is three-dimensional. Similar results hold true for all odd dimensions. More precisely, if $d\geq 5$ then we do not need to assume that $\widehat{\varphi}(0)=0$. If $d=1$ then we need $\nu\geq 2$, i.e. $\widehat{\varphi}$ must have a Taylor expansion near the origin starting with  a quadratic or higher order term. These conditions are closely related to the smoothness of the resolvent of the coupled system near the threshold $z=0$. As a rule of thumb, if the coupled resolvent is smooth enough near the threshold and  if the eigenvalue couples in a non-trivial way with the continuum, then not only will the systems have a bound state embedded at the threshold, but also its Heisenberg evolution will become orthogonal to the critical eigenprojection after diving in the continuum.

\item One can also allow the quantum dot to be modeled by $\mathbb{C}^N$ with $N>1$. Imagine a situation in which the initial instantaneous Hamiltonian has $N$ discrete eigenvalues and, during the adiabatic evolution, one eigenvalue dives in the continuum while the other $N-1$ stay away from it but can cross among themselves. Then one can prove that the `diving' eigenvalue is `lost' while the other $N-1$ return to the initial state.

\item This problem is closely related to the so-called `adiabatic pair creation', see \cite{N,PD,Th}.
Our model is much simpler than the one in \cite{PD} and allows for much more explicit dispersive estimates and a very detailed threshold analysis. At the same time, our model can capture quite different adiabatic behaviors depending on the properties of the coupling function $\varphi$. In spite of the simplicity of the model, the method we use is robust and can be generalized. An outlook on the Dirac and $N$-body Schr\"odinger case with $N\geq 2$ is presented in Section~\ref{sec: Dirac}. A detailed study will be done elsewhere. 

\item 
The case in which the bound state becomes a resonance at the threshold (in the sense that the full resolvent at criticality has an expansion around zero with a leading term of order $z^{-1/2}$) is still open.
\end{enumerate}
\end{remark}

\section{Proof of Theorem \ref{thm-horia}{\rm (i)}}\label{spectral-analysis}
The spectral properties close to the critical value have been studied for some different models in \cite{CJN,PD}. In this section we give a detailed spectral ana\-ly\-sis of the family of operators $H_{\tau}(E)$. If $\tau=0$ the operator $H_0(E)$ has $\sigma_{\rm ess}(H_0(E))=[0,\infty)$, and it is absolutely continuous. It has a simple eigenvalue at $E$, which is either isolated if $E<0$ or embedded in the continuous spectrum if $E\geq 0$. Since the perturbation is of finite rank we have $\sigma_{\rm ess}(H_{\tau}(E))=[0,\infty)$ also for $\tau\neq0$.

We recall that we always assume $\varphi\in L^2(\mathbb{R}^3)$. We write $r_0(z)=(-\Delta-z)^{-1}$ for $z\not\in[0,\infty)$.

We recall some results from~\cite{JK,JN2}.
We start with the asymptotic expansion of $r_0(z)$ which is a simple consequence of the Taylor formula for the kernel of $r_0(z)$.
To this end we change the variable to $\kappa = - i \sqrt{z}$ with the choice $\im\sqrt{z} > 0$ such that $\kappa^2=-z$.
The kernel of $r_0(-\kappa^2)$ is given by $e^{-\kappa \abs{x-y}}/(4\pi\abs{x-y})$.
We recall the following lemma whose proof can be found in~\cite{JK}. 
\begin{lemma}\label{lem: expansion of r_0}
Let $n\geq1$ and $w>n+\frac32$. Then
\begin{align*}
r_0(-\kappa^2)=\sum_{j=0}^n \kappa^jG_j+\mathcal{O}(\kappa^{n+1})
\end{align*}
as $\kappa\to0$ with $\re \kappa \geq 0$ in the norm topology of $\mathcal{B}(L^{2,w}(\mathbb{R}^3),L^{2,-w}(\mathbb{R}^3))$.
The expansion coefficients are operators given by the integral kernels
\begin{align*}
G_j(x,y)=\frac{(-1)^j}{4\pi j!}\abs{x-y}^{j-1}
\end{align*}
and have the mapping properties
\begin{align*}
G_0 &\in \mathcal{B}(L^{2,w_1}(\mathbb{R}^3),L^{2,-w_2}(\mathbb{R}^3))\quad \text{with}\quad w_1,w_2>\tfrac12 \quad \text{and}\quad w_1+w_2\geq2,\\
G_j &\in \mathcal{B}(L^{2,w_1}(\mathbb{R}^3),L^{2,-w_2}(\mathbb{R}^3))\quad \text{with}\quad w_1,w_2>j+\tfrac12.
\end{align*}
\end{lemma}

In the following we also use the notation $\ip{\cdot\,}{\,\cdot}$ for the inner product viewed as a duality between $L^{2,w}(\mathbb{R}^3)$ and $L^{2,-w}(\mathbb{R}^3)$.

There is an important relation between $G_0$ and $G_2$ given as follows:

\begin{proposition}\label{prop: relation between G_0 and G_2}
Assume $f,g\in L^{2,w}(\mathbb{R}^3)$ with $w>\frac52$ such that $\ip{f}{1}=0$ and $\ip{g}{1}=0$.
Then $G_0f,G_0g\in L^{2}(\mathbb{R}^3)$ and
$\ip{f}{G_2g}=-\ip{G_0f}{G_0g}$.
\end{proposition}
The proof can be found in~\cite{JK,JN2}. A related result is the following lemma:
\begin{lemma}\label{prop-hc1}
Let $\varphi$ be as in Assumption~\emph{\ref{cond: coupling 1}}. Then the map 
\begin{align*}
(-\infty,0]\ni x \mapsto \ip{\varphi}{r_0(x)\varphi}\in \mathbb{R}
\end{align*}
is continuously differentiable. 
Moreover, there exists a constant $C$ such that for all $\kappa>0$,
\begin{align}\label{horia9}
\vert \ip{\varphi}{[r_0(-\kappa^2)]^2\varphi}\vert \leq C,\quad\vert \ip{\varphi}{[r_0(-\kappa^2)]^3\varphi}\vert \leq C\kappa^{-1},\quad \vert \ip{\varphi}{[r_0(-\kappa^2)]^4\varphi}\vert \leq C\kappa^{-3}.
\end{align}
\end{lemma}
\begin{proof}
We start with the identity
\begin{align}\label{horia7'}
\ip{\varphi}{r_0(-\kappa^2)\varphi}=\int_{\mathbb{R}^6}\frac{e^{-\kappa\abs{x-y}}}{4\pi\abs{x-y}}\overline{\varphi(x)}\varphi(y)dxdy.
\end{align}
Differentiating on both sides with respect to $\kappa$ and using that \begin{equation*}
\int_{\mathbb{R}^3}\varphi(x)dx=(2\pi)^{3/2}\widehat{\varphi}(0)=0
\end{equation*}
we have
\begin{align}\label{horia7}
\ip{\varphi}{[r_0(-\kappa^2)]^2\varphi}&=\frac{1}{8\pi\kappa}\int_{\mathbb{R}^6}e^{-\kappa\abs{x-y}}\overline{\varphi(x)}\varphi(y)dxdy\nonumber \\
&=-\frac{1}{8\pi}\int_0^1 \int_{\mathbb{R}^6}e^{-t\kappa\abs{x-y}}\overline{\varphi(x)}\varphi(y)\abs{x-y}dxdydt.
\end{align}
In particular, setting $u=-\kappa^2$, this shows that the functions given by $f(u)=\ip{\varphi}{r_0(u)\varphi}$ and $f'(u)=\ip{\varphi}{[r_0(u)]^2\varphi}$ are continuous for all $u\leq 0$. 

The right hand side of \eqref{horia7} is $C^\infty$ as a function of $\kappa$ and all its derivatives are bounded. Differentiating once with respect to $\kappa$ on the left hand side we get
\begin{equation*}
4\kappa \ip{\varphi}{[r_0(-\kappa^2)]^3\varphi}.
\end{equation*}
Hence $\ip{\varphi}{[r_0(-\kappa^2)]^3\varphi}$ diverges at most like $\kappa^{-1}$. Differentiating twice on the left hand side of \eqref{horia7} yields
\begin{equation*}
4\ip{\varphi}{[r_0(-\kappa^2)]^3\varphi}+24\kappa^2 \ip{\varphi}{[r_0(-\kappa^2)]^4\varphi},
\end{equation*}
a quantity which must be bounded as a function of $\kappa$. Hence $\ip{\varphi}{[r_0(-\kappa^2)]^4\varphi}$ might diverge at most like $\kappa^{-3}$. 
\end{proof}

\noindent We need a representation of the resolvent of $H_{\tau}(E) = H_0(E) + \tau V$ with $V=\begin{bmatrix} 0 & \ket{\varphi}\\ \bra{\varphi} & 0 \end{bmatrix}$. We use the notation $R(z)=(H_{\tau}(E)-z)^{-1}$ omitting explicit mention of the dependence on $\tau$, $\varphi$, and $E$.

Now let $z \in \mathbb{C}$ with $\im (z) \neq 0$.
As a first step we rewrite the resolvent $R(z)$ using the notation $r_0(z) = (-\Delta - z )^{-1}$ and $R_0(z) = (H_0(E) - z )^{-1}$.
Note that $r_0(z)$ is an operator on $L^2(\mathbb{R}^3)$ while $R(z)$ and $R_0(z)$ are operators on $\mathcal{H} = L^2(\mathbb{R}^3) \oplus \mathbb{C}$.
The deduced representation of the resolvent $R(z)$ comes from~\cite{JN1} and is a variant of the Feshbach formula. The Feshbach formula should more precisely be called the Schur--Livsic--Feshbach--Grushin formula.

The perturbation is factored as $V=wUw^*$ where $w\colon \mathbb{C}^2\to\mathcal{H}$ and $U\colon \mathbb{C}^2\to\mathbb{C}^2$ are given by
\begin{align*}
w= \begin{bmatrix}
\ket{\varphi} & 0 \\ 0 & 1
\end{bmatrix}
\quad\text{and}\quad
U = \begin{bmatrix}
0 & 1 \\  1 & 0
\end{bmatrix}.
\end{align*}
With this notation we define the complex $2\times2$-matrix
\begin{align}\label{eq: definition of M}
M(z) = U + \tau w^*R_0(z)w = \begin{bmatrix}
\tau\ip{\varphi}{r_0(z)\varphi} & 1\\
1 & \tau (E-z)^{-1}
\end{bmatrix}.
\end{align}
Then for $\im z\neq0$ we have the representation
\begin{equation*}
R(z) = R_0(z) - \tau R_0(z)wM(z)^{-1}w^{\ast}R_0(z).
\end{equation*}
The Feshbach map is defined as
\begin{equation}\label{def F}
F(z,E) = E - z - \tau^2\ip{\varphi}{r_0(z)\varphi}.
\end{equation}
We have $\det M(z) = -(E-z)^{-1} F(z,E)$ and
\begin{align*}
M(z)^{-1}=\frac{1}{\det M(z)}
\begin{bmatrix}
\tau(E-z)^{-1} & -1 \\
-1 & \tau\ip{\varphi}{r_0(z)\varphi}
\end{bmatrix}.
\end{align*}
Using this notation and \eqref{def F} we can rewrite the resolvent as
\begin{align}
R(z) &= 
\begin{bmatrix}
r_0(z) & 0 \\
0 & 0
\end{bmatrix}+
\frac{1}{F(z,E)}
\begin{bmatrix}
\tau^2 r_0(z)\ket{\varphi}\bra{\varphi} r_0(z)  & -\tau  r_0(z)\ket{\varphi}\\
-\tau \bra{\varphi}r_0(z) & 1
\end{bmatrix}.
\label{eq: representation of R}
\end{align}

It is clear from \eqref{eq: representation of R} that any  isolated eigenvalue must be a zero of $F(z,E)$. Assume $z=x$ is real and negative. Then
\begin{equation*}
\frac{\partial F}{\partial x}(x,E)=-1-\tau^2\ip{r_0(x)\varphi}{r_0(x)\varphi}.
\end{equation*}
Thus the derivative is negative for all negative $x$. We note that $F(x,E)$ is positive if $x$ is sufficiently negative.

Let $g_0=\lim_{z\to0,\,z\in(-\infty,0)}\ip{\varphi}{r_0(z)\varphi}$. This limit exists, see Lemma \ref{prop-hc1}, and we have $g_0=\ip{\varphi}{G_0\varphi}>0$. Set $E_c=\tau^2g_0$. Here we must choose $\tau$ smaller than some constant $\tau_0$ such that $E_c < E_m$.  Then the function $F(z,E)$  has a unique simple negative zero at $\lambda(E)<0$ for $E<E_c$. 
It is determined by the equation $F(\lambda(E),E) = 0$. From the implicit function theorem and \eqref{def F} we obtain
\begin{align}\label{horia8}
&\lambda'(E)=\big( 1+\tau^2\ip{\varphi}{[r_0(\lambda(E))]^2\varphi}\big)^{-1},\notag\\
&\lim_{E\uparrow E_c}\lambda(E)=\lambda(E_c)=0,\quad {\rm and}\quad \lambda(E)\sim -(E_c-E).
\end{align}

Again from \eqref{eq: representation of R} we see that $\lambda(E)$ is a simple pole of $R(z)$ and Cauchy's residue theorem yields the orthogonal Riesz projection corresponding to it:
\begin{align}\label{P(E)}
P(E) = \frac{1}{1+\tau^2\ip{\varphi}{[r_0(\lambda(E))]^2\varphi}}
\begin{bmatrix}
\tau^2r_0(\lambda(E))\ket{\varphi}\bra{\varphi} r_0(\lambda(E))
&
-\tau r_0(\lambda(E))\ket{\varphi}\\
-\tau \bra{\varphi} r_0(\lambda(E))
&
1
\end{bmatrix}.
\end{align}
We see that $P(E)=\ket{\Psi(E)}\bra{\Psi(E)}$ where a normalized eigenvector is given by
\begin{align}\label{Psi(E)}
\Psi(E) = \frac{1}
{\sqrt{1+\tau^2\ip{\varphi}{r_0(\lambda(E))^2\varphi}}}
\begin{bmatrix}
-\tau r_0(\lambda(E))\ket{\varphi} \\ 1
\end{bmatrix}.
\end{align}

Now let us prove the norm estimate on $P(E)$ in \eqref{horia6}. When we differentiate in \eqref{P(E)} there are two terms which might have a singular behavior when $E$ goes to $E_c$. One of these terms is proportional to the scalar $\ip{\varphi}{[r_0(\lambda(E))]^3\varphi}$ and the other singular term contains the vector $[r_0(\lambda(E))]^2\ket{\varphi}$. Hence the singular factors in $\norm{P'(E)}$ are
\begin{equation*}
\abs{\ip{\varphi}{[r_0(\lambda(E))]^3\varphi}}\quad {\rm and}\quad \norm{[r_0(\lambda(E))]^2\ket{\varphi}}=\sqrt{\ip{\varphi}{[r_0(\lambda(E))]^4\varphi}}.
\end{equation*}
However, the estimates in \eqref{horia9} show that these terms can diverge at most like $\abs{\lambda(E)}^{-3/4}$. Together with \eqref{horia8} this imply the norm bound on $P'(E)$ in \eqref{horia6}.

Next we study in detail the critical resolvent, i.e. when $E=E_c$, in a neighborhood of $\lambda(E_c)=0$. Using Lemma~\ref{lem: expansion of r_0} we write
\begin{align}\label{eq: expansion of F}
\ip{\varphi}{r_0(-\kappa^2)\varphi} = g_0 + \kappa g_1 + \kappa^2 g_2 + \mathcal{O}(\kappa^3)
\end{align}
with $g_j = \ip{\varphi}{G_j\varphi}$ for $j=0,1,2$. Note that Assumption~\ref{cond: coupling 1} implies $\ip{\varphi}{1} = 0$ such that $g_1=0$.
 
We want to extract the singular part in the asymptotic expansion around $\kappa = 0$ from representation~\eqref{eq: representation of R}. From the expansion~\eqref{eq: expansion of F} we deduce
\begin{align*}
\lim_{\kappa\to 0}\frac{\kappa^2}{F(-\kappa^2,E_c)}=\frac{1}{1 - \tau^2 g_2}.
\end{align*}
By Proposition~\ref{prop: relation between G_0 and G_2} we have $g_2 = \ip{\Psi_c}{G_2\Psi_c} \leq 0$ such that $1-\tau^2 g_2 \geq 1$.
Therefore the above limit is positive.

In order to identify the projection corresponding to the point spectrum of $H_\tau(E_c)$ at $0$ we need to compute the weak limit 
\begin{equation*}
\lim_{\kappa\to 0}\kappa^2R(-\kappa^2).
\end{equation*}
From \eqref{eq: representation of R} we find that only the second term can give a non-zero contribution. Since $r_0(-\kappa^2)\ket{\varphi}$ in the Fourier space coincides with $\widehat{\varphi}(k)/(k^2+\kappa^2)$ and $\widehat{\varphi}(k)/k^2$ is an $L^2$-function we may write in the $L^2$-sense
\begin{equation*}
\frac{\widehat{\varphi}(k)}{k^2+\kappa^2}-\frac{\widehat{\varphi}(k)}{k^2} =-\frac{\kappa^2}{k^2+\kappa^2}\frac{\widehat{\varphi}(k)}{k^2}.
\end{equation*}
However, the right hand side converges to zero in the $L^2$-norm when $\kappa$ goes to zero (as a consequence of the Lebesgue dominated convergence theorem). In other words, we proved that $G_0\ket{\varphi}$ is actually an $L^2$-function and  $r_0(-\kappa^2)\ket{\varphi}$ converges strongly to $G_0\ket{\varphi}$. 
We conclude that asymptotically
\begin{align}\label{eq: asympt expansion of R}
\kappa^2 R(-\kappa^2)=P_c + o(1)\quad {\rm with}\quad 
P_c = \frac{1}{1-\tau^2g_2}\begin{bmatrix}
\tau^2 G_0\ket{\varphi} \bra{\varphi} G_0 & -\tau G_0\ket{\varphi}\\\
-\tau \bra{\varphi} G_0 & 1
\end{bmatrix}.
\end{align}
It follows from general spectral theory that $P_c$ is an orthogonal projection. 
From this expression we determine the normalized eigenvector
\begin{align}\label{eq: normalised eigenvector Psic}
\Psi_c = \frac{1}{\sqrt{1-\tau^2g_2}}
\begin{bmatrix}
-\tau G_0\ket{\varphi}\\ 1
\end{bmatrix}.
\end{align}
A comparison of this result with \eqref{P(E)} shows that $P_c$ is the norm limit of $P(E)$ when $E\uparrow E_c$ and the proof of \eqref{horia6} is finished. 

\section{Proof of Theorem \ref{thm-horia}{\rm (ii)}}

We want to estimate the difference between the Heisenberg evolved initial projection and the critical instantaneous eigenprojection, i.e.
\begin{equation*}
U_{\eta} ( s_c/\eta,-1/\eta )P(E(-1))U_{\eta} ^*( s_c/\eta,-1/\eta )-P_c.
\end{equation*}
We first give a proof which works for all $\nu\geq 1$ and afterwards we show how the error estimates can be improved if $\nu$ is large enough. 

\subsection{The case $\nu=1$}

The problem we need to circumvent is that the gap between the instantaneous eigenvalue $\lambda(E(s))$ and $0$ vanishes when $s\uparrow s_c$. The main idea (which is not new) is to use the group property of the unitary evolution and write the Heisenberg evolved projection as
\begin{equation*}
U_{\eta} ( s_c/\eta,s/\eta )\left [ U_{\eta} ( s/\eta,-1/\eta )P(E(-1))U_{\eta}^* ( s/\eta,-1/\eta )\right ]U_{\eta}^* ( s_c/\eta,s/\eta )
\end{equation*}
for some $s \in (-1,s_c)$ to be chosen later. Since $s<s_c$ we know that the instantaneous eigenvalue stays away from the rest of the spectrum and we have a minimal gap $g = -\lambda(E(s))\sim E_c-E(s)$. From the usual adiabatic theorem with a gap~(see~\cite{T2}) we know that there is a constant $C$ such that
\begin{align*}
\norm{ U_{\eta} ( s/\eta,-1/\eta )P(E(-1))U_{\eta}^* ( s/\eta,-1/\eta )-P(E(s))} \leq C\frac{\eta}{g^3}
\leq C\frac{\eta}{(E_c-E(s))^3}.
\end{align*}
This estimate is only useful when $s$ is sufficiently far from $s_c$, i.e. $E_c-E(s)\sim s_c-s$ is much larger than $\eta^{1/3}$. If $s$ gets closer to $s_c$ than $\eta^{1/3}$ we need a complementary estimate which is only useful when $E(s)$ is close to $E_c$ and then combine the two. 

The following proposition deals with the `near' region: 

\begin{proposition}\label{prop: evolution s_0 to s_c}
Let $\nu\geq 1$. 
There is a constant $C > 0$ such that for any macroscopic time $s \in [-1,s_c]$
\begin{align*}
\norm{ U_{\eta} ( s_c/\eta,s/\eta ) P(E(s)) U_{\eta}^* ( s/\eta,s_c/\eta ) - P_c} \leq C \abs{E(s)-E_c}^{1/4}.
\end{align*}
\end{proposition}

\begin{proof} We may assume that $s<s_c$. 
For $t \in [ s/\eta, s_c/\eta ]$ we define the projection
\begin{align*}
\widetilde{P} (t) = U_{\eta}^* ( t,s_c/\eta ) P(E(\eta t)) U_{\eta} (t,s_c/\eta).
\end{align*}
Using that $H_\tau(E(\eta t))$ commutes with $P(E(\eta t))$ we compute
\begin{align*}
\widetilde{P}'(t) = \eta E'(\eta t) U_{\eta}^* ( t,s_c/\eta ) P'(E(\eta t)) U_{\eta} (t,s_c/\eta).
\end{align*}
Thus
\begin{align*}
\widetilde{P} ( s/\eta ) = P_c - \int_{s/\eta}^{s_c/\eta} \widetilde{P}'(t) d t\quad {\rm and} \quad \norm{ \widetilde{P}( s/\eta ) - P_c } \leq \int_{E(s)}^{E_c} \norm{P'(E)} d E.
\end{align*}
From \eqref{horia6} we know that $\norm{P'(E)} \leq C ( E_c - E )^{-3/4}$ and, hence, after integration and using that $U_{\eta} ( s_c/\eta,s/\eta )=U_{\eta}^* ( s/\eta,s_c/\eta )$ the proof is complete.
\end{proof}

Combining the estimates for the far and the near region we show that there exists a constant $C$ such that for every $s<s_c$
\begin{equation*}
\norm{U_{\eta} ( s_c/\eta,-1/\eta )P(E(-1))U_{\eta} ^*( s_c/\eta,-1/\eta )-P_c}\leq C((s_c-s)^{1/4}+\eta/(s_c-s)^3).
\end{equation*}
Let us choose $s$ such that the two terms in the parenthesis become equal. Then we obtain $s_c-s=\eta^{4/13}$, hence 
\begin{equation*}
\norm{U_{\eta} ( s_c/\eta,-1/\eta )P(E(-1))U_{\eta} ^*( s_c/\eta,-1/\eta )-P_c}\leq C \eta^{1/13}.
\end{equation*}
In the general case $\nu\geq 1$ the exponent $\alpha(1)$ in \eqref{horia1} can be chosen to be $1/13$. 

\subsection{The case $\nu$ large}

The main idea behind the `standard' adiabatic theorem with a minimal gap $g>0$ is to find an operator $Q(E)$ which satisfies the commutator identity $i[H_\tau(E),Q(E)]=-P'(E)$. If this is achieved then 
the operator 
\begin{equation*}F_\eta(t)=U^*(t,-1/\eta)\bigl(P(E(\eta t))+\eta E'(\eta t)Q(E(\eta t))\bigr)U(t,-1/\eta)
\end{equation*}
obeys $\norm{F_\eta'(t)}\leq C(g) \eta^2$. Hence $F_\eta(t)-F_\eta(-1/\eta)$ is of order $C(g) \eta$ on any macroscopic time interval.

In general the norms of $Q(E)$ and $Q'(E)$ diverge when the gap goes to zero. We will show in the following that both $Q(E)$ and $Q'(E)$ remain bounded up to $E_c$ if $\nu$ is large enough, which will allow us to choose an exponent $\alpha(\nu)=1$. 

If $E<E_c$ we define the reduced resolvent
\begin{equation*}
R^\perp(z)=P^\perp(E)R(z)P^\perp(E).
\end{equation*}
We know that $R^\perp(z)$ has a removable singularity at $\lambda(E)$ and is analytic in a small disk centered at $\lambda(E)$ and 
\begin{align}\label{horia11}
R^\perp(\lambda(E))=\frac{1}{2\pi i}\int_{|z-\lambda(E)|=\abs{\lambda(E)}/2}\frac{1}{z-\lambda(E)}R(z)dz.
\end{align}
Then the operator we are looking for can be defined as
\begin{align*}
Q(E) = i R^\perp(\lambda(E)) P'(E)P(E)- i P(E)P'(E) R^\perp(\lambda(E)).
\end{align*}

We rewrite $Q(E)$ in order to see why both, $Q(E)$ and $Q'(E)$, are bounded up to the threshold provided $\nu$ is large enough. 
We will use \eqref{eq: representation of R} in order to simplify the expression in \eqref{horia11}. First we observe that
\begin{align*}
R(z) &= 
\begin{bmatrix}
r_0(z) & 0 \\
0 & 0
\end{bmatrix}+
\frac{1}{F(z,E)}\ket{\Phi(z)}\bra{\Phi(z)} \quad {\rm with}\quad \ket{\Phi(z)}=\begin{bmatrix}
-\tau r_0(z)\ket{\varphi}\\ 1
\end{bmatrix}.
\end{align*}
From \eqref{def F} we also have that
\begin{equation*}
m(z,E)=\frac{z-\lambda(E)}{F(z,E)}=-\big(\int_0^1 (1+\tau^2\ip{\varphi}{r_0^2((1-t)\lambda(E)+tz)\varphi}dt\big)^{-1}
\end{equation*}
is analytic for $\re(z)<0$. Then Cauchy's residue theorem applied to \eqref{horia11} yields
\begin{align}\label{horia12}
R^\perp(\lambda(E))=\begin{bmatrix}
r_0(\lambda(E)) & 0 \\
0 & 0
\end{bmatrix}+\left (m(z,E)\ket{\Phi(z)}\bra{\Phi(z)}\right )'|_{z=\lambda(E)}.
\end{align}
The first term in \eqref{horia12} is never bounded when $\lambda(E)$ tends to zero. The crucial observation is that the product $R^\perp(\lambda(E))P'(E)$ has finite rank and is  better behaved, see also \eqref{P(E)} for the expression of $P(E)$. The worst singularity of $R^\perp(\lambda(E))P'(E)$ appears in the vector $r_0^3(\lambda(E))\ket{\varphi}$ and the same holds for $Q(E)$. When we analyze $Q'(E)$ its worst singularity appears in the vector $r_0^4(\lambda(E))\ket{\varphi}$. 

Now the question is to find the minimal $\nu$ which insures that 
\begin{equation*}
\norm{r_0^4(\lambda(E))\ket{\varphi}}=\sqrt{\ip{\varphi}{r_0^8(\lambda(E))\varphi}}
\end{equation*}
remains bounded when $\lambda(E)$ approaches zero. We find
\begin{equation*}
\ip{\varphi}{r_0^8(\lambda(E))\varphi}=\int_{\mathbb{R}^3}\frac{\abs{\widehat{\varphi}(k)}^2}{(k^2+\abs{\lambda(E)})^8}dk\leq \int_{\mathbb{R}^3}\frac{\abs{\widehat{\varphi}(k)}^2}{\abs{k}^{16}}dk.
\end{equation*}
Therefore a sufficient condition for the right hand side to be bounded is $\nu\geq 7$. 

\section{Proof of Theorem \ref{thm-horia}{\rm (iii)}}\label{subsec: evolution oc}

From now on we will deal with the Hamiltonian $H_{\tau}(E)$ when $E_c\leq E$. Consider the standard set of generalized eigenfunctions of $-\Delta$ given by $(2\pi)^{-3/2}e^{i \ip{k}{x}}$ and the corresponding set of generalized eigenfunctions of the decoupled Hamiltonian  $H_0(E)$ are given by
\begin{equation*}
\ket{\Psi_k^{(0)}}=\begin{bmatrix}(2\pi)^{-3/2}e^{i \ip{k}{x}} \\ 0\end{bmatrix}.
\end{equation*}
Using the Lippmann--Schwinger equation we can define the corresponding generalized eigenfunctions of $H_\tau(E)$, denoted by $\ket{\Psi_k}$ and given by 
\begin{equation}\label{horia12'}
\ket{\Psi_k}=\ket{\Psi_k^{(0)}}-\tau (H_\tau(E)-k^2-i0_+)^{-1}\ket{\varsigma}\ip{\varphi}{\Psi_k^{(0)}}.
\end{equation}
This implies that (see also \eqref{eq: representation of R})
\begin{equation}\label{horia13}
\ip{\varsigma}{\Psi_k}=-\tau \ip{\varsigma}{(H_\tau(E)-k^2-i0_+)^{-1}|\varsigma} \overline{\widehat{\varphi}(k)}=-\tau \overline{\widehat{\varphi}(k)}/{F(k^2+i0_+,E)}.
\end{equation}
From \eqref{def F}, replacing $z$ by $r^2+i0_+$ with $r>0$, we obtain
\begin{equation}\label{horia14}
F(r^2+i0_+,E)=E-E_c-r^2-\tau^2\lim_{\epsilon\downarrow 0}\int_{\mathbb{R}^3}\abs{\widehat{\varphi}(k)}^2
\Bigl(\frac{1}{k^2-r^2-i\epsilon}-\frac{1}{k^2}\Bigr)dk.
\end{equation}
Let $\kappa_\epsilon$ be the principal branch of $(r^2+i\epsilon)^{1/2}$. Then we have
\begin{equation}\label{horia15}
\int_{\mathbb{R}^3}\abs{\widehat{\varphi}(k)}^2
\Bigl( \frac{1}{k^2-r^2-i\epsilon}-\frac{1}{k^2}\Bigr)dk
=\int_{\mathbb{R}^6}\overline{\varphi(x)}\varphi(y)\frac{e^{i\kappa_\epsilon \abs{x-y}}-1}{4\pi\abs{x-y}}dxdy.
\end{equation}

In this section we need an extra assumption on $\varphi$ besides Assumption~\ref{cond: coupling 1}:

\begin{assumption}\label{cond: coupling 2}
Consider the expression $F(r^2+i0_+,E)$ in~\eqref{horia14} with $E>E_c$. We assume that there exists $E_a\in (E_c,E_m]$ such that $\inf_{r \in \mathbb{R}} \abs{ F(r^2+i0_+,E_a) } \geq c$ for some constant $c > 0$.
\end{assumption}
This assumption together with \eqref{eq: representation of R} imply
 that $H_\tau(E_a)$ has purely absolutely continuous spectrum equal to $[0,\infty)$. Let us comment on the class of functions $\varphi$ that satisfy this condition. Introducing \eqref{horia15} in \eqref{horia14} and taking the limit $\epsilon\downarrow 0$ we obtain
\begin{align}\label{horia20}
F(r^2+i0_+,E)=&E-E_c-r^2-\tau^2\int_{\mathbb{R}^6}\overline{\varphi(x)}\varphi(y)\frac{\cos(r\abs{x-y})-1}{4\pi\abs{x-y}}dxdy\nonumber \\
&-i\tau^2
\int_{\mathbb{R}^6}\overline{\varphi(x)}\varphi(y)\frac{\sin(r\abs{x-y})}{4\pi\abs{x-y}}dxdy.
\end{align}
Let $E>E_c$. If $r=0$ the real part of the above expression equals $E-E_c>0$. The real part is negative if $r$ is larger than some critical value. If $\tau$ is small enough then there is exactly one value $r_E\sim \sqrt{E-E_c}>0$ for which the real part equals zero. For Assumption~\ref{cond: coupling 2} to hold we need to make sure that the imaginary part of $F(r^2+i0_+,E)$ is not zero on a neighborhood of $r_E$. Applying the Sokhotski--Plemelj formula in \eqref{horia14}
we get
 \begin{align}\label{horia21}
\im F(r^2+i0_+,E)&=-\frac{\tau^2 r}{2}\int_{\mathbb{S}^2}\abs{\widehat{\varphi}(r,\omega)}^2d\omega\nonumber \\
&=-\tau^2
\int_{\mathbb{R}^6}\overline{\varphi(x)}\varphi(y)\frac{\sin(r\abs{x-y})}{4\pi\abs{x-y}}dxdy.
\end{align}
It follows that it suffices to have $\widehat{\varphi}(k)$ different from zero on a ball around $k=0$ not including the origin. For example, if $\nabla \widehat{\varphi}(0)\neq 0$ and if $\tau$ is small enough then  $r_E$ is also small and 
\begin{equation*}
\int_{\mathbb{S}^2}\abs{\widehat{\varphi}(r_E,\omega)}^2d\omega=Cr_E^2 +\mathcal{O}(r_E^4)\neq 0.
\end{equation*}
A similar argument works if some coefficient of the $n$-th order Taylor expansion at $k=0$ is different from zero while all coefficients up to order  $n-1$ are zero. We then have
\begin{equation*}
\int_{\mathbb{S}^2}\abs{\widehat{\varphi}(r_E,\omega)}^2d\omega=Cr_E^{2n} +\mathcal{O}(r_E^{2n+2})\neq 0. 
\end{equation*}
Another convenient scenario is when $\widehat{\varphi}$ is spherically symmetric and different from zero outside $k=0$. In this case, all its Taylor coefficients at $k=0$ may be zero. 

In the rest of this section we will prove that the survival probability of the critical eigenvector after being Heisenberg evolved between $s_c/\eta$ and $s'_c/\eta$ goes to zero with $\eta$ if the `dispersive' Assumption~\ref{cond: coupling 2} holds true. 
Due to the fact that the adiabatic theorem holds up to the threshold (see \eqref{horia1}), it is enough to show
\begin{equation}\label{horia22}
\lim_{\eta\downarrow 0}\abs{\ip{\Psi_c}{U_{\eta} (s_c'/\eta, s_c/\eta ) \Psi_c}}^2=0.
\end{equation}

To this end we will compare the true time-evolution generated by $H_\tau(E(\eta t))$ with the one generated by the time-independent, purely absolutely continuous Hamiltonian $H_\tau(E_a)$ where $E_a$ is overcritical and obeys the dispersive condition. We know that there exists at least one $s_a\in (s_c,s'_c)$ such that $E_a=E(s_a)$. In order to simplify the notation we set $H_a = H_\tau(E(s_a))$ and $t_a = s_a/{\eta}$.

The following propagation estimate will play a key role:
\begin{proposition}\label{prop: propagation estimate}
There exists a constant $C > 0$ such that for all $t\in\mathbb{R}$,
\begin{align}\label{eq: improved propagation estimate}
\abs{ \ip{\varsigma}{e^{ - i t H_a} \varsigma}} \leq C ( 1 + \abs{t} )^{-5/2}.
\end{align}
\end{proposition}
\begin{proof}
We assume $\varphi \in L^{2,w} (\mathbb{R}^3)$ for a $w \geq w_0 \geq 1$.
The $w_0$ will be specified below.
Using the resolvent representation~\eqref{eq: representation of R} for $E=E_a>E_c$ we have 
\begin{equation*}
\ip{\varsigma}{R(z) \varsigma} = \frac1{F(z,E_a)}.
\end{equation*}
The boundary values $F(\lambda \pm i0,E_a)$ exist for all $\lambda \geq 0$ and
$
F(\lambda \pm i0,E_a) = E_a - \lambda - \tau^2 \ip{\varphi}{r_0(\lambda \pm i0) \varphi}
$. 
We have for $n \geq 0$, $w_0 > n + \frac12$, and $z \in \mathbb{C} \setminus [0,\infty)$ that
\begin{equation*}
\frac{d^n}{dz^n}\ip{\varphi}{ r_0(z) \varphi} = \mathcal{O} \big(\abs{z}^{-(n+1)/2}\big) \quad \text{as} \ \abs{z} \to \infty,
\end{equation*}
see e.g. \cite[Theorem 8.1]{JK}.
Using this result we see that there are $C_1, C_2 >0$ such that
\begin{equation*}
\abs{F(z,E_a)} \geq C_1 \abs{z} \quad \text{for} \ \abs{z} \geq C_2\ \text{with}\ \im(z) \neq 0
\end{equation*}
and the estimate extends to the boundary values. We have
\begin{equation*}
\im \frac1{F(\lambda+i0,E_a)} = - \tau^2 \frac{\im \ip{\varphi}{r_0(\lambda+i0) \varphi}}{\abs{F(\lambda+i0,E_a)}^2}.
\end{equation*}
The choice of $E_a$ implies that the spectral density
$
\rho(\lambda) = \frac1{\pi} \im \ip{\varsigma}{R(\lambda+i0) \varsigma}
$
is well defined for $\lambda \geq 0$.
Recall that Assumption~\ref{cond: coupling 2} is supposed to hold.
The results above show that there is a constant $C > 0$ such that
\begin{equation*}
\abs{\rho(\lambda)} \leq C \min \{ 1, \lambda^{-5/2} \}, \quad \lambda > 0.
\end{equation*}
This implies that we have the representation
\begin{equation*}
\ip{\varsigma}{ e^{- it H_a} \varsigma} = \int_0^{\infty} e^{it\lambda} \rho(\lambda) d\lambda
\end{equation*}
where the integral is well defined.
If we assume $w_0 > n + \frac12$ for some integer $n \geq 1$ then we have that $\rho$ has $n$ continuous derivatives.
Furthermore, for some $\lambda_0$ sufficiently large, 
\begin{equation}\label{eq: arne1}
\abs{\rho^{(n)}(\lambda)} \leq C \lambda^{-(n+5)/2},\quad \lambda \geq \lambda_0.
\end{equation}
Let $\delta > 0$.
We now introduce a function $\chi \in C_0^{\infty}(\mathbb{R})$ such that $\chi(\lambda) = 1$ for $\lambda \in [0,\delta]$ and $\chi(\lambda) = 0$ for $\lambda \geq 2\delta$.
Define $\rho_<(\lambda) = \chi(\lambda) \rho(\lambda)$ and $\rho_>(\lambda) = (1-\chi(\lambda))\rho(\lambda)$.

It follows from~\eqref{eq: arne1} that $\rho_>^{(3)} \in L^1(\mathbb{R})$ such that
\begin{equation*}
\int_0^{\infty} e^{-it\lambda} \rho_>(\lambda) d\lambda = \mathcal{O}(t^{-3})\quad \text{as}\ t \to \infty,
\end{equation*}
see e.g. \cite[Lemma 10.1]{JK}.

Using Lemma~\ref{lem: expansion of r_0} and $\ip{\varphi}{1} = 0$ we have
$\rho_< (\lambda) = c\lambda^{3/2} + \mathcal{O}(\lambda^2)$ as $\lambda \downarrow 0$.
The asymptotic expansion in Lemma~\ref{lem: expansion of r_0} can be differentiated and we have
\begin{equation*}
\rho_<'(\lambda) = c \lambda^{1/2} + \mathcal{O}(\lambda)\quad \text{as}\ \lambda \downarrow 0.
\end{equation*}
We can then use \cite[Lemma 10.2]{JK} to conclude that
\begin{equation*}
\int_0^{\infty} e^{-it\lambda} \rho_<(\lambda) d \lambda = \mathcal{O}(t^{-5/2})\quad \text{as}\ t \to \infty.
\end{equation*}
Thus we have proved the propagation estimate
\begin{equation*}
\abs{\ip{\varsigma}{e^{-itH_a} \varsigma}} \leq C(1+t)^{-5/2},\quad t \geq 0.
\end{equation*}
In order to obtain this result we need to have $w_0 > 9/2$.
\end{proof}

\subsection{A fundamental identity} 

We compare the time-evolution of the time-dependent Hamiltonian $H(\eta t)=H_\tau(E(\eta t))$ with the one generated by the Hamiltonian $H_a$.
Recall that $t_a = s_a/\eta$ and set $\varepsilon_a(t)=E_a-E(\eta t)$. 
The Dyson equation yields
\begin{align*}
U_{\eta} (t_a,t_c) = e^{-i(t_a-t_c)H_a} + i \int_{t_c}^{t_a} \varepsilon_a(u) e^{-i(t_a-u)H_a} \ket{\varsigma} \bra{\varsigma} U_{\eta}(u,t_c) du
\end{align*}
for the true time-evolution from $t_c = s_c/\eta$ to $t_a$.
Moreover, replacing $t_c$ by $t_c' = s_c'/\eta$ in this formula and taking the adjoint we obtain 
\begin{align*}
U_{\eta} (t_c',t_a) = U_{\eta}^* (t_a,t_c') = e^{-i(t_c'-t_a)H_a} + i \int_{t_a}^{t_c'} \varepsilon_a(v) U_{\eta} (t_c',v) \ket{\varsigma} \bra{\varsigma} e^{-i(v-t_a)H_a} dv.
\end{align*}
Using these two formulae as well as the group property we obtain an identity which will play a fundamental role in what follows:
\begin{align}\label{horia23}
U_{\eta} (t_c',t_c) &= U_{\eta} (t_c',t_a) U_{\eta} (t_a,t_c)\notag \\
&= e^{ - i ( t_c' - t_c ) H_a}
+ i \int_{t_c}^{t_a} \varepsilon_a(u) e^{-i (t_c'-u)H_a} \ket{\varsigma} \bra{\varsigma} U_{\eta} (u,t_c) d u\notag\\
&\phantom{=.} + i \int_{t_a}^{t_c'} \varepsilon_a(v) U_{\eta} (t_c',v) \ket{\varsigma} \bra{\varsigma} e^{ -i(v-t_c)H_a} d v\notag\\
&\phantom{=.} - \int_{t_a}^{t_c'} \int_{t_c}^{t_a} \varepsilon_a(u) \varepsilon_a(v) U_{\eta}(t_c',v) \ket{\varsigma} \ip{\varsigma}{e^{-i(v-u)H_a} \varsigma} \bra{\varsigma} U_{\eta}(u,t_c) du dv.
\end{align}
\subsection{Proof of the limit~\eqref{horia22} }

We use an $\varepsilon/2$-argument to prove \eqref{horia22}. Let $0<\varepsilon < 1$.
We can always construct an approximating vector $\Psi_{\varepsilon}$ such that $\norm{\Psi_c - \Psi_{\varepsilon}} < \frac{\varepsilon}{10}$ and $\Psi_{\varepsilon}$ is smooth in the spectral representation of $H_a$ and compactly supported away from the only threshold of $H_a$, $z=0$. 
We have that $\vert \ip{\Psi_c}{U_{\eta} (t_c',t_c) \Psi_c} - \ip{\Psi_{\varepsilon}}{U_{\eta} (t_c',t_c) \Psi_{\varepsilon}} \vert < \frac{\varepsilon}{2}$ uniformly in $\eta$. 
Thus it suffices to show that the overlap $\abs{\ip{\Psi_{\varepsilon}}{U_{\eta} (t_c',t_c) \Psi_{\varepsilon}}}$ goes to zero with $\eta$.

From \eqref{horia23} we obtain
\begin{align*}
\abs{\ip{\Psi_{\varepsilon}}{U_{\eta} (t_c',t_c) \Psi_{\varepsilon}}} &\leq \abs{\ip{\Psi_{\varepsilon}}{e^{ - i ( t_c' - t_c ) H_a} \Psi_{\varepsilon}}}\notag\\
&\phantom{=} + \int_{t_c}^{t_a} \varepsilon_a (v) \abs{\ip{\Psi_{\varepsilon}}{e^{ - i ( t_c' - v ) H_a} \varsigma}} d v\notag\\
&\phantom{=} + \int_{t_a}^{t_c'} \varepsilon_a (u) \abs{\ip{\varsigma}{e^{ - i ( u - t_c ) H_a} \Psi_{\varepsilon}}} d u \notag\\
&\phantom{=} + \int_{t_a}^{t_c'} \int_{t_c}^{t_a} \varepsilon_a (u) \varepsilon_a (v)  
\abs{\ip{\varsigma}{e^{ - i ( u-v ) H_a} \varsigma}} d v d u.
\end{align*}
We will show that each term on the right hand side, denoted in the order of appearance by $I_1$, $I_2$, $I_3$, and $I_4$, vanishes for $\eta \downarrow 0$.

Recall that the microscopic times depend on $\eta$ since $t_c = s_c/\eta$, $t_c' = s_c/\eta$, and $t_a = \frac{s_a}{\eta}$ where $s_c$, $s_c'$, and $s_a$ are fixed.
Thus the lengths of the intervals $[ t_c, t_c' ]$, $[ t_a,t_c' ]$, and $[ t_c, t_a ]$  are of order ${\eta}^{-1}$.

The scalar product appearing in the first term $I_1$ can be rewritten as the Fourier transform of a $C_0^\infty ((0,\infty))$-function. Hence it decays faster than any power of $\eta$ by the usual `integration by parts' argument.

In order to treat the second term it is crucial to observe that  $v \in [t_c, t_a ]$ such that $t_c'-v$ is always of order $\eta^{-1}$. A similar `integration by parts' argument as for $I_1$ yields that $I_2$ also decays faster than any power of $\eta$. 

The third term is similar to the second one if we note that $u\in [t_a,t_c']$. Hence $u-t_c$ is always of order $\eta^{-1}$.

We continue with the last term, $I_4$, which needs further analysis. 
Since $E(\cdot) \in C^1$ we have $\abs{E(s_a) - E(s)} \leq C \abs{s - s_a}$.
This estimate and Proposition~\ref{prop: propagation estimate} imply that there exists a constant $C_1 > 0$ such that
\begin{align*}
\abs{I_4} &\leq C_1 \eta^{2} \int_{t_a}^{t_c'} \int_{t_c}^{t_a} \abs{ u - t_a} \abs{v - t_a} ( 1 + \abs{u-v} )^{-5/2} d v d u\\
&= C_1 \eta^{2} \int_0^{t_c' - t_a} \int_{t_c - t_a}^0 \abs{u} \abs{v} ( 1 + \abs{u - v} )^{-5/2} d v d u.
\end{align*}
Since $u \geq 0 \geq v$, $1+u \leq 1+u-v$, and $1-v \leq 1+u-v$ we obtain
\begin{equation*}
( 1 + u - v )^{- 5/2}\leq ( 1+u )^{-5/4}( 1-v )^{-5/4}
\end{equation*}
and then
\begin{align*}
\abs{I_4}
&\leq C_1 \eta^{2} \int_0^{{(s_c' -s_a)}/{\eta}} \int_{{(s_c - s_a)}/{\eta}}^0 u (-v) ( 1 + u - v )^{-5/2} d v d u\\
&\leq C_1 \eta^{2}  \int_0^{{(s_c' -s_a)}/{\eta}} u ( 1+u )^{-5/4} d u  \int_0^{{(s_a - s_c)}/{\eta}} v ( 1+v )^{-5/4} d v \\
&\leq C_1 \eta^{2} \int_0^{{(s_c' -s_a)}/{\eta}} ( 1+u )^{-1/4} d u  \int_0^{{(s_a - s_c)}/{\eta}} ( 1+v )^{-1/4} d v .
\end{align*}
For $0 < \eta \leq \min \{s_c' -s_a,s_a - s_c\}$ we obtain
\begin{align*}
\int_0^{{(s_a - s_c)}/{\eta}} ( 1+v )^{-1/4} d v \leq \tfrac43( 1+{(s_a-s_c)}/{\eta})^{3/4} \leq C\eta^{-3/4},
\end{align*}
and an analogous estimate also holds for for the $u$ integral.
Finally, we conclude that there is a constant $C_2 >0$ such that 
$\abs{I_4} \leq C_2 \eta^{\frac12}$. 
Hence all four terms go to zero as $\eta \downarrow 0$ which concludes the proof of \eqref{horia22}.

\section{Proof of Theorem \ref{thm-horia}{\rm (iv)}}\label{abc}

We assume there exists $\delta>0$ such that $\widehat{\varphi}(k)=0$ for $\abs{k}\leq \delta$. Let $\chi_\delta(k)$ denote the characteristic function of the closed ball $B_\delta(0)$ and define an orthogonal projection $\pi_\delta$ in $L^2(\mathbb{R}^3)$ by
\begin{equation*}
\pi_\delta f=\mathcal{F}^{-1}(\chi_\delta \widehat{f}).
\end{equation*}
We extend this projection to $\mathcal{H}$ by setting $\Pi_\delta=\pi_\delta\oplus 0$. Then $\Pi_\delta$ commutes with the full Hamiltonian $H_\tau(E)$, see \eqref{eq: Hamiltonian}. The Hamiltonian is block-diagonal with respect to the decomposition induced by $\Pi_\delta$. The block $\Pi_\delta H_\tau(E) \Pi_\delta$ can be identified with the bounded operator $-\Delta_<$ given by
\begin{equation*}
-\Delta_<f=\mathcal{F}^{-1}(k^2\widehat{f})\quad \text{for} \quad \widehat{f}\in L^2(B_\delta(0)),
\end{equation*}
while the other block
\begin{align}\label{horia30}
h_\tau(E) = \begin{bmatrix} -\Delta_> & \tau \ket{\varphi}\\ \tau \bra{\varphi} & E \end{bmatrix}\quad {\rm with} \quad -\Delta_>f= 
\mathcal{F}^{-1}(k^2\widehat{f}),\quad \widehat{f}\in L^2(\mathbb{R}^3\setminus B_\delta(0)),
\end{align}
acts in $\bigl({\pi_\delta}^\perp L^2(\mathbb{R}^3)\bigr)\oplus \mathbb{C}$.

The spectrum of $H_\tau(E)$ is the union of the spectrum of $-\Delta_<$, i.e. the interval $[0,\delta^2]$, and the spectrum of $h_\tau(E)$. Since the adiabatic time-evolution also commutes with $\Pi_\delta$ it is factorized. 

If $E<\delta^2$ then the operator $h_0(E)$ has purely absolutely continuous spectrum in $[\delta^2,\infty)$ and a discrete bound state with energy $E$. For $\tau\neq 0$ the essential spectrum of $h_\tau(E)$ is still $[\delta^2,\infty)$. The crucial observation is that if we replace $r_0(z)$ with $(-\Delta_>-z)^{-1}$ in formula~\eqref{eq: representation of R} then we obtain  $(h_\tau(E)-z)^{-1}$. Moreover, $h_\tau(E)$ has discrete bound states if for given $E < \delta^2$ the `reduced' Feshbach map $F_{\delta}(\cdot,E)\colon(-\infty,\delta^2) \to \mathbb{R}$,
\begin{equation}\label{horia31}
F_\delta(x,E) = E - x - \tau^2 \int_{\abs{k}\geq \delta} \frac{\abs{\widehat{\varphi}(k)}^2}{k^2 - x}  d k,
\end{equation}
vanishes for some $x<\delta^2$.
The function $F_\delta(\cdot,E)$ is strictly decreasing from $+\infty$ and reaches negative values when $x\geq E$. Thus there is a unique discrete bound state $\lambda(E)$ such that 
\begin{equation*}
F_\delta(\lambda(E),E)=0\quad {\rm and}\quad \lambda(E)< E <\delta^2.
\end{equation*}
The conclusion is that we effectively work with a (reduced) system which always has a minimal gap. Hence the usual adiabatic theorem can be applied without problems in the subspace $\bigl({\pi_\delta}^\perp L^2(\mathbb{R}^3)\bigr)\oplus \mathbb{C}$. Since the instantaneous eigenvector always remains in this subspace and the full Heisenberg evolution is factorized, the proof of \eqref{horia4} follows.  

\section{Proof of Theorem \ref{thm-horia}{\rm (v)}}\label{subsec: oc-stable}

We go back to Assumption \ref{cond: coupling 1}. Then $\widehat{\varphi}(k)/\abs{k}^\nu$ is bounded (and non-zero) in a neighborhood of $k=0$. 

Let $\delta>0$ and let $\chi_\delta$ denote the characteristic function of $B_\delta(0)$. We define the cut-off function $\varphi_\delta$ by
\begin{align}\label{horia32}
\widehat{\varphi}_\delta(k)=(1-\chi_\delta(k))\widehat{\varphi}(k).
\end{align} 
By construction, $\widehat{\varphi}_\delta$ vanishes if $\abs{k}\leq \delta$. We define the Hamiltonian with cut-off
\begin{equation*}
H_{\tau}^\delta(E) = \begin{bmatrix} -\Delta & 0\\ 0 & E \end{bmatrix}
+ \tau\begin{bmatrix} 0 & \ket{\varphi_\delta}\\ \bra{\varphi_\delta} & 0 \end{bmatrix}.
\end{equation*}
We know from the previous section that $H_{\tau}^\delta(E)$ has a simple eigenvalue $\lambda_\delta(E)$ if $E<\delta^2$, but in the following we are interested in showing that this instantaneous eigenvalue survives for larger values of $E$ if $\delta$ is small. 

Set $s_1 = s_c-\alpha^{-1} \delta^{5/2}$ for some $\alpha >0$. Since $E(s_1)<E_c$ we know that the operator without a cut-off, i.e. $H_{\tau}(E(s_1))$, has a discrete and isolated negative eigenvalue  $\lambda(E(s_1))$. We have already seen that $\lambda'(E)$ is positive and bounded near $E<E_c$. Furthermore,
 \begin{align}\label{horia33}
\lambda(E(s_1))\sim s_1-s_c=-\alpha^{-1} \delta^{5/2}
\end{align} 
for small $\delta>0$.

\begin{lemma}\label{lemma-hc-2}
For sufficiently small $\alpha$ the operator $H_{\tau}^\delta(E(s_1))$ has a unique negative eigenvalue $\lambda_\delta(E(s_1))$ corresponding to an eigenprojection $P_\delta(E(s_1))$ and
\begin{equation*}
P_\delta(E(s_1))-P(E(s_1))=\mathcal{O}(\alpha)
\end{equation*}
uniformly in $\delta<1$. 
\end{lemma}
\begin{proof} We know that $H_{\tau}^\delta(E(s_1))$ can have at most one eigenvalue which must be the unique zero of the function in \eqref{horia31} if there is an eigenvalue. The difference 
\begin{equation*}
V_\delta=H_{\tau}(E)-H_{\tau}^\delta(E) = \tau\begin{bmatrix} 0 & \ket{\varphi-\varphi_\delta}\\ \bra{\varphi-\varphi_\delta} & 0 \end{bmatrix}
\end{equation*}
has the asymptotic behavior
\begin{align}\label{horia34}
\norm{V_\delta}&\sim \norm{\varphi-\varphi_\delta}_{L^2(\mathbb{R}^3)}\sim \delta^{\nu+3/2}.
\end{align}
In the norm estimate we used \eqref{horia32} and that $\widehat{\varphi}(k)$ behaves like $\abs{k}^\nu$ near zero. Since $\nu\geq 1$, $\norm{V_\delta}$ goes at least like $\delta^{5/2}$. 

We now use regular perturbation theory. If $z$ is at a distance of order $\alpha^{-1}\delta^{5/2}$ from the spectrum of $H_{\tau}(E(s_1))$ then $z$ is also in the resolvent set of $H_{\tau}^\delta(E(s_1))$ and the norm of $(H_{\tau}^\delta(E(s_1))-z)^{-1}$ is of order  $\alpha \delta^{-5/2}$.  Let $\Gamma_\delta$ be a positively oriented circle centered at $\lambda(E(s_1))$ and of radius $ \abs{\lambda(E(s_1))}/2\sim \alpha^{-1} \delta^{5/2}$. Then $\Gamma_\delta$ is in the resolvent set of $H_{\tau}^\delta(E(s_1))$. Finally, the second resolvent identity leads to
\begin{equation*}
\norm{\frac{1}{2\pi i}\int_{\Gamma_\delta}\bigl( (z-H_{\tau}^\delta(E(s_1)))^{-1}-(z-H_{\tau}(E(s_1)))^{-1}\bigr) dz }\leq C \alpha
\end{equation*}
for some constant $C>0$. This shows that, for $\alpha$ small enough, the Riesz integral 
\begin{equation*}\frac{1}{2\pi i}\int_{\Gamma_\delta} \bigl(z-H_{\tau}^\delta(E(s_1))\bigr)^{-1}dz
\end{equation*}
defines a rank one orthogonal projection which is $P_\delta(E(s_1))$. Moreover, it implies that there is an eigenvalue of $H_{\tau}^\delta(E(s_1))$ located inside $\Gamma_\delta$.
\end{proof}

\begin{lemma}\label{lemma-hc-3}
Let $s_2=s_c+\alpha \delta^2$. There is a $\delta_0<1$ such that the operator $H_{\tau}^\delta(E(s))$ for $0<\delta < \delta_0$ has a bound state $\lambda_\delta(E(s))$ for every $s\in[s_1, s_2]$. We have 
\begin{equation*}P_\delta(E(s))-P_\delta(E(s_1))=\mathcal{O}(\alpha)
\end{equation*}
uniformly in $\delta<\delta_0$ for every $s\in[s_1, s_2]$. 
\end{lemma}
\begin{proof}
In the system with cut-off the projections  live in a reduced Hilbert space.   
Choose $\delta_0(\alpha)$ such that $\alpha^{-1}\delta^{5/2}<\alpha \delta^2$ for any $\delta<\delta_0$.

We know from Lemma \ref{lemma-hc-2} that $\lambda_\delta(E(s_1))$ is negative and at a distance of order $\alpha\delta^{2}$ from zero. The essential  spectrum of the reduced Hamiltonians $h_\tau^\delta(E(s))$ (see \eqref{horia30}) is always $[\delta^2,\infty)$. 

The difference $E(s)-E(s_1)\sim s-s_1$ is of order $\alpha \delta^2$ for all $s\in [s_1,s_2]$. So, there is $C>0$ such that
\begin{equation*}
\norm{ h_{\tau}^\delta(E(s))-h_{\tau}^\delta(E(s_1))}\sim E(s)-E(s_1)\leq C \alpha \delta^2.
\end{equation*}

We again use regular perturbation theory. If $z$ is at a distance of order $\delta^2/2$ from the spectrum of $h_{\tau}^\delta(E(s_1))$ and if $\alpha$ is small enough then $z$ is also in the resolvent set of $h_{\tau}^\delta(E(s))$ by the usual Neumann series argument. Now choose a circle centered at $\lambda_\delta(E(s_1))$ and with radius $\delta^2/2$. Both resolvents, i.e. at times $s$ and $s_1$, diverge at most like $\delta^{-2}$ on this circle, which has circumference $\pi\delta^2$. The same type of argument as in Lemma \ref{lemma-hc-2} based on  Riesz integrals and  the second resolvent identity give
\begin{equation*}
\norm{P_\delta(E(s))-P_\delta(E(s_1))}\leq C \delta^{-4} \delta^2 \alpha\delta^2=C\alpha
\end{equation*}
for some constant $C>0$. Hence $P_\delta(E(s))$ has rank one uniformly in $\delta<\delta_0$ if $\alpha$ is small enough. This also shows that $\lambda_\delta(E(s))$ lies inside the above circle and therefore $\lambda_\delta(E(s))\leq \delta^2/2$ for every $s\leq s_2$. 
\end{proof}

\begin{lemma}\label{lemma-hc-4}
Let $U_{\eta,\delta}(t,t')$ denote the Heisenberg evolution generated by $H_{\tau}^\delta(E(\eta t))$. Then
\begin{equation*}
U_{\eta,\delta}(s_2/\eta,s_1/\eta)P_\delta(E(s_1))U_{\eta,\delta}^*(s_2/\eta,s_1/\eta)=P_\delta(E(s_2))+\mathcal{O}((s_2-s_1)/\delta^2).
\end{equation*}
\end{lemma}

\begin{proof}
Consider the projection valued map given by
\begin{equation*}
\Pi_\eta(t)=U_{\eta,\delta}^*(t,s_2/\eta)P_\delta(E(\eta t))U_{\eta,\delta}(t,s_2/\eta),\quad t\in [s_1/\eta,s_2/\eta].
\end{equation*}
Differentiating and integrating back we obtain
\begin{equation*}
\norm{\Pi_\eta(s_2/\eta)-\Pi_\eta(s_1/\eta)}\leq (s_2-s_1)\sup_{s\in [s_1,s_2]}\norm{P_\delta'(E(s))}.
\end{equation*}
We know that there is a minimal gap of order $\delta^2/2$ between the instantaneous eigenvalue $\lambda_\delta(E(s))$ and the rest of the (essential) spectrum of the reduced Hamiltonian $h_{\tau}^\delta(E(s))$. Then the norm of $P_\delta'(E(s))$ is bounded by the inverse of the minimal effective gap. This proves the lemma. 
\end{proof}

\begin{lemma}\label{lemma-hc-5}
Let $U_{\eta}(t,t')$ denote the Heisenberg evolution generated by $H_{\tau}(E(\eta t))$ and let $U_{\eta,\delta}(t,t')$ denote the Heisenberg evolution generated by $H_{\tau}^\delta(E(\eta t))$. Then there is $C>0$ such that 
\begin{equation*}
\norm{U_{\eta}(t,t')-U_{\eta,\delta}(t,t')}\leq C \abs{t-t'}\delta^{\nu+3/2}.
\end{equation*}
\end{lemma}

\begin{proof}
Consider the Dyson integral identity which relates the two unitaries. Since the difference between their generators is $V_\delta$, see \eqref{horia34}, the difference between the two evolution operators is bounded by the length of the time interval times the norm of the perturbation.   
\end{proof}

\begin{lemma}\label{lemma-hc-6} 
Fix $\alpha\in (0,1)$. Define $\eta_0(\alpha)$ to be a solution of
\begin{equation*}
\alpha^{-1} \bigl(\eta_0^{{2}/{(2\nu+7)}}\bigr)^{5/2}= \alpha \bigl(\eta_0^{{2}/{(2\nu+7)}}\bigr)^{2}.
\end{equation*}
Then for every $\eta<\eta_0(\alpha)$ and $s\in (s_c,s_c+\alpha \eta^{{4}/{(2\nu+7)}}]$,
\begin{align}\label{hans1}
U_{\eta}(s/\eta,s_c/\eta)P_cU_{\eta}^*(s/\eta,s_c/\eta) =P_c+ \mathcal{O}(\alpha).
\end{align}
\end{lemma}
\begin{proof}
Each $s$ can be written as $s=s_2=s_c+\alpha\delta^2$ with $\delta\in (0,\eta^{\frac{2}{2\nu+7}}]$. The definition of $\eta_0(\alpha)$ implies that 
\begin{equation*}
\alpha^{-1}\delta^{5/2}\leq \alpha\delta^2.
\end{equation*}

In the subcritical regime the adiabatic theorem holds up to the critical point, see \eqref{horia1}. So
\begin{equation*} P_c=U_{\eta}(s_{c}/\eta,s_1/\eta)P(E(s_1))U_{\eta}^*(s_{c}/\eta,s_1/\eta) +o_\eta(1)
\end{equation*}
holds uniformly in $s_1\in [-1,s_c)$. Let us choose $s_1=s_c-\alpha^{-1}\delta^{5/2}$. Our lemma is proved if we can show
\begin{align}\label{horia35}
U_{\eta}(s_{2}/\eta,s_1/\eta)P(E(s_1))U_{\eta}^*(s_{2}/\eta,s_1/\eta) =P_c+ \mathcal{O}(\alpha),
\end{align}
provided $\eta<\eta_0(\alpha)$.

By Lemma~\ref{lemma-hc-5} with $t=s_2/\eta$ and $t'=s_1/\eta$ the error due to replacing the evolution without cut-off by the one with cut-off is of order $\alpha \eta^{-1}\delta^{\nu+7/2}=\alpha$. Hence \eqref{horia35} is equivalent with
\begin{equation}\label{horia36}
U_{\eta,\delta}(s_{2}/\eta,s_1/\eta)P(E(s_1))U_{\eta,\delta}^*(s_{2}/\eta,s_1/\eta) =P_c+ \mathcal{O}(\alpha),
\end{equation}
provided $\eta<\eta_0(\alpha)$. Lemma~\ref{lemma-hc-2} implies that, up to an error of order $\alpha$, we may replace $P(E(s_1))$ by $P_\delta(E(s_1))$ on the left hand side of \eqref{horia36}. Then Lemma~\ref{lemma-hc-4} implies that, again up to an error of order $\alpha$, the left hand side of \eqref{horia36} can be replaced by $P_\delta(E(s_2))$. Lemma~\ref{lemma-hc-3} shows that $P_\delta(E(s_2))$ is actually close to $P_\delta(E(s_1))$. Using again Lemma \ref{lemma-hc-2}  we can replace $P_\delta(E(s_1))$ by $P(E(s_1))$. Finally, the subcritical projection $P(E(s_1))$ is close in norm to the critical projector $P_c$, see \eqref{horia6}, and the circle is closed. 
\end{proof}

Now the proof of \eqref{thm: main-3} is a direct consequence of Lemma~\ref{lemma-hc-6} by taking the expectation with respect to $\Psi_c$ on both sides of \eqref{hans1}.

\section{Outlook}\label{sec: Dirac}
We shall briefly explain why we expect that our method of proving zero survival probability \eqref{thm: main-4} can work in other quite different settings.
Before that, let us recall two crucial aspects of our proof: 

\begin{enumerate}
\item From the resolvent expression in \eqref{eq: representation of R} we see that when we restrict it to the small sample, its asymptotic expansion around the threshold does not contain the linear term in $\kappa$ due to the condition $\langle\varphi |1\rangle =0$. In particular, this implies that the eigenvalue remains an embedded eigenvalue at the threshold, no resonances are possible. We  want to note that \cite{PD} has a similar implicit assumption, see equation (3) in \cite{PD} which states that no resonances are allowed at the threshold. For such a thing to happen, their external `critical' potential needs to obey a number of algebraic conditions, see Theorem 1.1 (iv) and (v) in \cite{Kl}.

\item From Proposition \ref{prop: propagation estimate} we see that the fixed {\it overcritical} evolution operator $e^{itH_a}$ has a dispersive behavior like  $t^{-\alpha}$ with $\alpha >2$ when restricted to the small sample, a property which plays a decisive role in estimating the last term in \eqref{horia23}. 
Compared to \cite{PD}, we do not need any information on how the instantaneous generalized {\it critical} eigenfunctions behave near the threshold \cite{P}. 
\end{enumerate}

We now switch to the Dirac case and consider the free Dirac operator for a particle of unit mass $
D_0 := - i \vec{\alpha} \cdot \vec{\nabla} + \beta$ 
acting in the Hilbert space $L^2(\mathbb{R}^3) \otimes \mathbb{C}^4$.
Here $\alpha_1,\alpha_2,\alpha_3$ and $\beta$ are the Dirac matrices \cite{Kl,Th} and $\vec{\alpha} := (\alpha_1,\alpha_2,\alpha_3)$.
The spectrum of $D_0$ is purely absolutely continuous and equals $(-\infty,-1] \cup [1,\infty).$  

Let us add to $D_0$ a radially symmetric potential $\mu(s)V(\abs{x})\otimes \mathbf{1}_4$ where $ V(\abs{x})\geq 0$ is compactly supported and  $\mu(s)\geq 0$ varies smoothly with the macroscopic  time parameter $s\in[-1,0]$. We are interested in the operator $$D(s):=D_0+\mu(s) V(\abs{x})\otimes \mathbf{1}_4$$ when $\mu(s)$ varies in such a way that at some critical time $s=s_c$ the operator $D(s_c)$ has an embedded eigenvalue at $+1$.

The orbital angular momentum $\vec{L} = (L_1,L_2,L_3)$ is given by $L_k := (i \epsilon_{kmn} x_m \partial/\partial_{x_n}) \otimes \mathbf{1}_4$ where $\epsilon_{kmn}$ is the totally antisymmetric tensor and $\mathbf{1}_4$ is the identity matrix on $\mathbb{C}^4$.
$\vec{L}^2$ and $L_3$ have a set of joint eigenprojections.
We choose their representation in polar coordinates, $p_{\ell,m} := \mathbf{1}_{\mathrm{rad}} \otimes \ket{Y_{\ell,m}}\bra{Y_{\ell,m}}$ where $\mathbf{1}_{\mathrm{rad}}$ is the identity operator on the radial component and $Y_{\ell,m}$ are the spherical harmonics for the angular components with quantum numbers $\ell \geq 0$ and $\abs{m} \leq \ell$.
The spin is $\vec{S} = (S_1,S_2,S_3)$ with $S_k := \mathbf{1}_{L^2} \otimes \frac12\begin{bmatrix} \sigma_k &0 \\ 0 & \sigma_k \end{bmatrix}$ where $\mathbf{1}_{L^2}$ is the identity operator on $L^2(\mathbb{R}^3)$ and $\sigma_1,\sigma_2,\sigma_3$ the Pauli spin matrices.
The component $S_3$ has two doubly degenerate eigenvalues $\pm 1/2$ corresponding to four eigen-bispinors $\psi_b\in \mathbb{C}^4$. 

The total angular momentum is $\vec{J} = \vec{L} + \vec{S}$.
Note that $\vec{J}^2$ has the eigenvalues $j(j+1)$ where the quantum number $j$ can be any positive half-integer.
Both $\vec{J}^2$ and $J_3$ commute with $D(s)$.
We now fix $j\geq 3/2$ and denote by $\Pi_{j}$ the projection on the corresponding eigenspace of $\vec{J}^2$. The only two allowed angular momentum quantum numbers are $\ell_{\pm} =j\pm 1/2$, and both of them are different from zero.   We have
\begin{equation*}
\Pi_{j}=\mathbf{1}_{\mathrm{rad}}\otimes \sum_{a,a'=\pm }\sum_{|m_{a}|\leq \ell_a} \sum_{|m_{a'}|\leq \ell_{a'}}\sum_{b,b'=1}^4 C_{\ell_a, \ell_{a'}, m_a, m_{a'},b,b'}\; \ket{Y_{\ell_a,m_a}}\bra{Y_{\ell_{a'},m_{a'}}} \otimes \ket{\psi_b}\bra{\psi_{b'}}
\end{equation*}
where the constants are normalizing numerical coefficients.  The operator $D(s)$ commutes with $\Pi_j$ at all times.

The first important result to be proved is that in the appropriate topology, the reduced resolvent $\Pi_j(D(s)-z)^{-1} \Pi_j$ has `a more regular asymptotic expansion than usual' near the two thresholds $\pm 1$ when $\mu(s)$ is not `critical'. In particular, this amounts to proving that when $|z|<1$, the two non-smooth terms proportional with $(1-|z|)^{\pm 1/2}$  are absent; let us exemplify this with the free Dirac operator. Using formula \cite[Eq.~2.1--5]{Kl} where $A_1^{\pm} = 1/(4\pi) \ket{1}\bra{1} \otimes \left( \beta \pm \mathbf{1}_4 \right)$ we have
\begin{align*}
\left( D_0 - z \right)^{-1} &= A_0^{\pm} + A_1^{\pm } \left(1-z^2\right)^{1/2} + A_2^{\pm} (1-z^2) + A_3^{\pm} (1 - z^2)^{3/2} + \ldots
\end{align*}
for which we see that $\Pi_j A_1^{\pm }\Pi_j=0$ due to the presence of $Y_{\ell_a,m_a}$ with $\ell_a\neq 0$. This roughly corresponds to the condition $\langle \varphi |1\rangle$ imposed in our current paper. The improved behavior of the reduced resolvent near thresholds makes it possible for the eigenvalues corresponding to a $j\geq 3/2$ to remain embedded eigenvalues at thresholds. By varying $\mu(s)$, eigenvalues can come out from one threshold,  cross the gap and hit the other threshold, with their eigenfunctions lying in the range of $\Pi_j$ at all times.  

The second important ingredient to be done is a refinement of \eqref{horia23} which demands for a detailed study of the dispersive properties of $\Pi_j e^{-itD_a}P_{\rm ac}(D_a)\Pi_j$ in the appropriate topology, for which we need a decay like $t^{-\alpha}$ with $\alpha>2$. Here $D_a=D(s_a)$ corresponds to a fixed overcritical time $s_a$ when the instantaneous bound state has disappeared.

The case $j=1/2$ is more special and does not enter our framework. The reduced resolvent $\Pi_{1/2} (D(s)-z)^{-1}\Pi_{1/2}$ cannot be made `regular enough' unless some other conditions are imposed on $V$ in order to prevent the appearance of resonances at thresholds. We recall that \cite{PD} also fails to cover the resonant case, see condition (3) in \cite{PD}.

Let us end this outlook by mentioning that our method will also be extended to certain $N$-body systems where $N\geq 2$ (with their center of the mass removed). While the case  $N=2$ has only one threshold and its spectral and scattering analysis is fairly well understood, when $N\geq 3$ the threshold structure is much more complicated and also the needed dispersive estimates are considerably more challenging.

\subsection*{Acknowledgment}

We would like to thank Peter Pickl for very interesting and fruitful discussions. 
HC, AJ, and HKK were partially supported by the Danish Council of Independent Research $\mid$ Natural Sciences, Grant DFF--4181-00042.
GN was partially supported by the VELUX Visiting Professor Program.

\end{document}